\documentclass[11pt]{article}

\usepackage{amsmath, amssymb}

\usepackage{amsthm}

\usepackage{ifthen}

\usepackage{ifpdf}

\ifpdf
\usepackage[pdftex]{graphicx}
\usepackage[pdftex,pdfborder={0 0 0}]{hyperref}
\else
\usepackage[dvips]{graphicx}
\fi

\usepackage[utf8x]{inputenc}
\usepackage{algorithmic, algorithm}

\usepackage{amsmath, amssymb}

\usepackage{amsthm}

\usepackage{ifthen}

\usepackage{ifpdf}

\usepackage[small,bf]{caption}
\usepackage{mathtools}
\usepackage{comment}
\usepackage{color}
\usepackage[]{todonotes}
\usepackage{xspace}
\usepackage{float}

\ifthenelse{\isundefined{\NoGenerateLetterSize}}
{
\ifpdf
\setlength{\pdfpagewidth}{8.5in}
\setlength{\pdfpageheight}{11in}
\else

\fi
}
{}

\newcommand{\Comment}[1]{\ignorespaces}

\long\def\LongVersion#1\LongVersionEnd{#1}
\long\def\ShortVersion#1\ShortVersionEnd{}

\LongVersion 
\LongVersionEnd 

\ShortVersion 
\usepackage[letterpaper, margin=0.8in, left=0.8in, right=0.8in]{geometry}
\ShortVersionEnd 

\ShortVersion 
\usepackage{times}
\ShortVersionEnd 

\newtheorem{theorem}{Theorem}[section]
\newtheorem{lemma}[theorem]{Lemma}
\newtheorem{observation}[theorem]{Observation}
\newtheorem{corollary}[theorem]{Corollary}

\theoremstyle{definition}

\newtheorem*{definition}{Definition}
\theoremstyle{plain}

\newcommand\BigO{\mathcal{O}}

\newcommand{\Sec}[0]{Sect.}
\newcommand{\Appendix}[0]{Appendix}
\newcommand{\Figure}[0]{Fig.}
\newcommand{\Equation}[0]{Eq.}

\LongVersion 
\newenvironment{MathMaybe}[0]
{\begin{displaymath}\ignorespaces}
{\end{displaymath}\ignorespacesafterend}
\LongVersionEnd 
\ShortVersion 
\newenvironment{MathMaybe}[0]
{\let\oldenspace\endspace \def\enspace{} \begin{math}\ignorespaces}
{\end{math}\ignorespacesafterend \def\enspace\oldenspace}
\ShortVersionEnd 

\newcommand\lc[1]{\ensuremath{{#1}_{_{\!L\!}}}}
\newcommand\hc[1]{\ensuremath{{#1}_{_{\!H\!}}}}
\newcommand\wb{\ensuremath{\psi}}
\newcommand\Wb{\ensuremath{\Psi}}
\newcommand\gl{\ensuremath{\lambda}}
\newcommand\eff{\ensuremath{\eta}}
\newcommand\lw{\ensuremath{\Lambda}}
\newcommand\mc{\ensuremath{\xi}\xspace}
\newcommand\ratio{\ensuremath{\rho}}
\newcommand{\grt}[0]{H}

\newcommand\cor\sim
\newcommand\Greedy{\textsc{Greedy}\xspace}
\providecommand{\dist}[1]{\ensuremath{\mathrm{dist}_{#1}}}

\newcommand{\Integers}{\mathbb{Z}}
\newcommand{\PoA}{\ensuremath{\mathrm{PoA}}\xspace}

\newcommand{\PoS}{\ensuremath{\mathrm{PoS}}\xspace}

\begin{document}

\title{The Price of Matching Selfish Vertices \\[8pt]
\large{or: How much love is lost because our governments do not arrange
marriages?}}

\author{Yuval Emek
\and
Tobias Langner
\and
Roger Wattenhofer \\[8pt]
Computer Engineering and Networks Lab (TIK), ETH Zurich, Switzerland \\
\texttt{\{yemek,\,langnert,\,wattenhofer\}@tik.ee.ethz.ch}
}

\date{}

\begin{titlepage}

\maketitle

\begin{abstract}
We analyze the setting of minimum-cost perfect matchings with selfish
vertices through the price of anarchy (\PoA{}) and price of stability (\PoS{})
lens.
The underlying solution concept used for this analysis is the Gale-Shapley
stable matching notion, where the preferences are determined so that each
player (vertex) wishes to minimize the cost of her own matching edge.
\end{abstract}

\noindent
\textbf{Keywords:}
minimum-cost perfect matching,
stable matching,
price of anarchy,
price of stability,
metric costs,
$\alpha$-stability.

\renewcommand{\thepage}{}
\end{titlepage}

\pagenumbering{arabic}

\section{Introduction}
Studying the impact of selfish players has been a major theoretical computer
science success story in the last decade (see, e.g., the 2012 G\"{o}del Prize
\cite{KoutsoupiasP09,RoughgardenT02,NisanR01}).
In particular, much effort has been invested in quantifying how the
efficiency of a system degrades due to the selfishness of its players.
The most notable notions in this context are the \emph{price of anarchy
(\PoA{})}~\cite{KoutsoupiasP09,Papa01} and the \emph{price of stability
(\PoS{})}~\cite{SSM03,ADK+08}, comparing the best possible outcome to the
outcome of the worst (\PoA{}) or best (\PoS{}) solution with selfish players.
Selfishness in this regard is usually captured by the \emph{Nash equilibrium}
solution concept, where no player can benefit from a unilateral deviation.

The players considered in the current paper are identified with the vertices
of a complete (or complete bipartite) weighted graph;
our goal is then to analyze the \PoA{} and \PoS{} of \emph{minimum-cost
perfect matchings}, where the efficiency of an outcome (a matching incident to
all vertices) is measured in terms of the sum of edge weights (a.k.a.\ costs).
Since unilateral deviations do not make sense in a matching setting, we
replace the Nash equilibrium solution concept with that of the
\emph{Gale-Shapley stable matching} notion~\cite{GS62}, where no two unmatched
players (strictly) prefer each other over their current matching partners,
defining the preferences so that each player wishes to minimize the weight of
the matching edge on which she is incident.

It is not difficult to show that a stable perfect matching always exists in a
complete (or complete bipartite) weighted graph with an even number of
vertices (cf.\ \cite{ABE+09} or the proof of
Lemma~\ref{lemma:UnstableEdgeCreation} in the current paper).
Yet, a simple example shows that in general, the situation is hopeless
(unbounded \PoA{} and \PoS{}):
Let $G$ be a complete graph on four nodes $u_1, u_2, v_1, v_2$ with edge
weights $w(u_1, v_1) = w(u_2, v_2) = 1$, $w(u_1, u_2) = \varepsilon$ for some
small $\varepsilon > 0$, and $w(v_1, v_2) = w(u_1, v_2) = w(u_2, v_1) = W$ for
some large $W$.
Then, the optimal perfect matching matches $u_i$ to $v_i$ for $i = 1,
2$ with a cost of $2$, whereas the unique stable matching (and any reasonable
approximation thereof) must match $u_1$ to $u_2$, and hence also $v_1$ to
$v_2$ which incurs a large cost.

The problem becomes much more interesting if we restrict ourselves to
\emph{metric} instances, namely, graphs with edge weights that obey the
\emph{triangle inequality} (or its bipartite counterpart).
Such instances correspond to settings where the players' preferences are
biased towards players of a similar type, e.g., when the players prefer to be
matched to players 
of a geographical proximity,
with a similar taste in film and music, or
with a similar appreciation for coriander.
Indeed, we establish an upper bound of $\BigO (n^{\log (3 / 2)})$ on the
\PoA{} and \PoS{} of minimum-cost perfect matchings in metric graphs with $2
n$ vertices, where $\log (3 / 2) \simeq 0.58$, and show that this is
asymptotically tight.\footnote{
In this paper, $\log x$ denotes the logarithm of $x$ to the base of $2$.
}
The somewhat unattractive polynomial dependency on $n$ raises the following
question:
How does \PoS{} improve once the Gale-Shapley stability is relaxed to
\emph{$\alpha$-stability}, where two unmatched vertices deviate from the
current matching only if both improve their costs by a factor greater than
$\alpha \geq 1$?
(Observe that since, by definition, every stable matching is also
$\alpha$-stable, this question is irrelevant in the context of \PoA{} that can
only increase by such a relaxation.)
We answer this question by establishing an asymptotically tight trade-off,
showing that with respect to $\alpha$-stable matchings, \PoS{} improves to
$\Theta (n^{\log \left( 1 + \frac{1}{2 \alpha} \right)})$;
in particular, taking $\alpha = O (\log n)$ yields a constant \PoS{}.
All our results hold for both simple and bipartite metric graphs.

\paragraph{Related work.}
Finding a maximum matching in a graph is among the most extensively studied
problems in combinatorial optimization.
Edmonds presented the first poly-time algorithm for the unweighted version of
the problem as well as a solution for finding a maximum-weight matching in
weighted graphs~\cite{Edmonds65,Edmonds65a} and initiated a long and fruitful
line of work on this problem
\cite{HK73,MV80,RT81,ABM+91,GT91,Motwani94,MS04}.
Reducing the minimum-weight perfect matching problem in complete graphs to the
maximum-weight matching problem is trivial.

In the stable matching setting, originally introduced by Gale and
Shapley~\cite{GS62}, each node is equipped with a totally ordered list of
\emph{preferences} on the other nodes.
Gale and Shapley showed that in the bipartite (\emph{marriage}) variant, a
stable matching always exists, and in fact, can be computed by a simple
poly-time algorithm.
In contrast, the all-pairs (\emph{roommates}) variant does not necessarily
have a solution.
Both variants of the stable matching problem admit a plethora of literature;
see, e.g., the books of Knuth~\cite{Knuth76}, Gusfield and Irving~\cite{GI89},
and Roth and Sotomayoror~\cite{RS90}.

Sometimes, the nodes' preferences are associated with real
\emph{costs} so that each preference list is sorted in order of increasing (or
non-decreasing if ties are allowed) costs.
This setting gives rise to the problem of computing a \emph{minimum-cost}
stable matching (a generalization of the \emph{egalitarian} stable matching
problem).
Irving et al.~\cite{ILG87} and Feder~\cite{Feder92} designed poly-time
algorithms for the bipartite variant of this problem;
the NP-hardness of the all-pairs variant was established by
Feder~\cite{Feder92} who also showed that the problem admits a
$2$-approximation.

The results discussed so far apply to arbitrary preference lists, where the
nodes' preferences exhibit no intrinsic correlations.
Several approaches have been taken towards introducing some consistency in the
preference lists~\cite{Knuth76,NH91,IMS08}.
Most relevant to the current paper is the approach of Arkin et
al.~\cite{ABE+09} who studied the \emph{geometric} stable roommate problem,
where the nodes correspond to points in a Euclidean space and the preferences
are given by the sorted distances to the other points.
They showed that in the geometric setting, a stable matching always exists and
that it is unique if the nodes' preferences exhibit no ties.
These results easily generalize to arbitrary metric spaces.
Arkin et al.\ also introduced the notion of an \emph{$\alpha$-stable} matching
for $\alpha \geq 1$ --- which is central to the current paper --- where nodes
are only willing to switch to a new match if they can improve over their
current partner by more than an $\alpha$-factor.

From a game theoretic perspective, it is interesting to point out that the
algorithm of Gale and Shapley is not \emph{incentive compatible}, namely, a
strategic player will not necessarily cooperate with this algorithm when
probed for her preferences.
In fact, Roth~\cite{Roth82} showed that there does not exist a stable marriage
algorithm under which, it is a dominant strategy for all players to be
truthful about their preferences.
We do not consider the issue of incentive compatibility in the current paper
(it is not even clear how this is defined in a weighted undirected graph).

The price of anarchy was introduced by Koutsoupias and
Papadimitriou~\cite{KoutsoupiasP09,Papa01} and since then has become a
cornerstone of algorithmic game theory.
The price of stability was first studied by Schulz and
Stier Moses~\cite{SSM03}, while the term itself was coined by Anshelevich et
al.~\cite{ADK+08}.
Since their introduction, the price of anarchy and the price of stability have
been extensively analyzed in diverse settings such as
selfish
routing~\cite{RoughgardenT02,Roughgarden03,ADK+08,SuriTZ07,AwerbuchAE05,ChristodoulouK05a,ChristodoulouK05b},
network formation
games~\cite{Vetta02,AnshelevichDTW03,AlbersEEMR06,ChenR06,AndelmanFM07},
job scheduling~\cite{KoutsoupiasP09,CzumajV02,KoutsoupiasMS03,AwerbuchART06},
and resource allocation~\cite{JohariT04,Roughgarden06}.


\section{Setting and Preliminaries}
\label{section:Model}
Consider a graph $G$ with vertex set $V(G)$ and edge set $E(G)$.
Each edge $e \in E(G)$ is assigned with a positive real \emph{weight} $w(e)$.
Unless stated otherwise, the graphs mentioned in this paper have $2 n$
vertices, $n \in \Integers_{> 0}$, and they are either complete
($|E(G)| = {2 n \choose 2}$)
or complete bipartite
($V(G) = U_1 \cup U_2$, $|U_1| = |U_2| = n$ and $|E(G)| = n^{2}$).
We say that the complete (or complete bipartite) graph $G$ is \emph{metric} if
$w(x, y) = \dist{G}(x, y)$
for every edge $(x, y) \in E(G)$, where $\dist{G}(x, y)$ denotes the
\emph{distance} between $x$ and $y$ in $G$ with respect to the edge weights
$w(\cdot, \cdot)$.

A \emph{matching} is a subset $M \subseteq E(G)$ of the edges such that every
vertex in $V(G)$ is incident to at most one edge in $M$.
The matching is called \emph{perfect} if every vertex in $V(G)$ is incident to
exactly one edge in $M$, which implies that $|M| = n$ as $|V(G)| = 2 n$.
For a perfect matching $M$ and a vertex $x \in V(G)$, we denote by $M(x)$ the
unique vertex $y \in V(G)$ such that $(x, y) \in M$.
Unless stated otherwise, all matchings mentioned hereafter are assumed to be
perfect.
(Perfect matchings clearly exist in a complete or complete bipartite graph
with an even number of vertices.)
Given an edge subset $F \subseteq E(G)$, we define the \emph{cost} of $F$ as
the total weight of all edges in $F$, denoted by
$c(F) = \sum_{e \in F} w(e)$;
in particular, the cost of a matching is the sum of its edge weights.

\begin{definition}[$\alpha$-Stable Matching]
Consider some (perfect) matching $M \subseteq E(G)$ and some real number $\alpha \geq
1$.
An edge $(u,v) \notin M$ is called \emph{$\alpha$-unstable} with respect to
$M$ if $\alpha \cdot w(u, v) < \min \{ w(u, M(u)), w(v, M(v)) \}$.
Otherwise, the edge is called \emph{$\alpha$-stable}.
A matching $M$ is called \emph{$\alpha$-stable} if it does not admit any
$\alpha$-unstable edge.
We will omit the parameter $\alpha$ and call edges as well as matchings just
\emph{stable} or \emph{unstable} whenever $\alpha$ is clear from the context
or the argumentation holds for every choice of $\alpha$. 
\end{definition}

Let $M^{*}$ denote a certain (perfect) matching $M$ that minimizes $c(M)$.
For simplicity, in what follows, we restrict our attention to complete (rather
than complete bipartite) graphs, although all our results hold also for the
complete bipartite case.

\begin{definition}[Price of Anarchy]
The \emph{price of anarchy} of a graph $G$, denoted by $\PoA(G)$, is defined as 
$\PoA(G) =
\max \{ c(M) / c(M^{*}) : M \text{ is a stable matching} \}$.
Let
$\PoA(2 n) =
\sup \{ \PoA(G) : G \text{ is metric, } |V(G)| = 2 n \}$.
\end{definition}

\begin{definition}[$\alpha$-Price of Stability]
The \emph{$\alpha$-price of stability} of $G$, denoted by $\PoS_{\alpha}(G)$, is
defined as
$\PoS_{\alpha}(G) =
\min \{ c(M) / c(M^{*}) : M \text{ is an $\alpha$-stable matching} \}$.
Let
$\PoS_{\alpha}(2 n) =
\sup \{ \PoS_{\alpha}(G) : G \text{ is metric, } |V(G)| = 2 n \}$.
Unless stated otherwise, when the parameter $\alpha$ is omitted, we refer to
the case $\alpha = 1$.
\end{definition}

\section{Price of Anarchy}
\label{section:PoA}

Our goal in this section is to establish the following theorem.

\begin{theorem} \label{theorem:PoA}
The \PoA{} of minimum-cost perfect matchings in metric graphs with $2 n$
vertices is $\Theta (n^{\log (3 / 2)})$.
\end{theorem}

Theorem~\ref{theorem:PoA} is established via a series of reductions,
essentially showing that $\PoA(2 n)$ is realized by weighted line graphs,
namely, metric graphs that can be embedded isometrically into the real line.
Following that, we introduce a family of weighted line graphs with \PoA{} of
$\Theta (n^{\log (3 / 2)})$ and show that no other weighted line graph admits
higher \PoA{}.
It is interesting to point out that this family of weighted line graphs was
first introduced by Reingold and Tarjan~\cite{RT81} for the analysis of a
greedy algorithm approximating the minimum-cost perfect matching problem in
metric graphs (with no stability considerations).

\begin{definition}[Matching Configuration]
A \emph{matching configuration (MC)} $\mc = (G, M^*, M)$ consists of a metric
graph $G$, a minimum-cost matching $M^*$, and a stable matching $M$ on $G$.
The \emph{ratio} of $\mc$ is defined as $\ratio(\mc) := c(M)/c(M^*)$.
\end{definition}

Observe that the definition of a MC \mc implies a
collection $\mathcal{A}(\mc)$ of alternating cycles in the symmetric
difference $M \oplus M^{*}$;
the cycles in $\mathcal{A}(\mc)$ are referred to hereafter as the alternating
cycles \emph{exhibited} by $\mc$.
We say that $\mc$ is \emph{spanned} by the cycles in $\mathcal{A}(\mc)$ if
each vertex of $G$ belongs to an alternating cycle in $\mathcal{A}(\mc)$.
Clearly, graphs with $2$ vertices admit a single (perfect) matching, hence
$\PoA(2) = 1$, so in what follows, it suffices to consider MCs on $2 n$
vertices for $n > 1$.
The following lemma states that it also suffices to consider MCs spanned by a
single alternating cycle.

\begin{lemma} \label{lemma:SingleAlternatingCycle}
For every MC $\mc = (G, M^{*}, M)$ on $2 n$ vertices, there exists a MC
$\hat{\mc}$ on $2 n'$ vertices, $1 < n' \leq n$, spanned by a single
alternating cycle such that
$\ratio(\hat{\mc}) \geq \ratio(\mc)$.
\end{lemma}
\begin{proof}
Since $\mathcal{A}(\mc) = \emptyset$ implies $\ratio(\mc) = 1$, we may
assume hereafter that $|\mathcal{A}(\mc)| \geq 1$, so let $A$ be an
alternating cycle in $\mathcal{A}(\mc)$ that maximizes the ratio
$c(M_{A}) / c(M^{*}_{A})$,
where $M_{A}$ and $M^{*}_{A}$ are the matchings $M^{*}$ and $M$, respectively,
restricted to the edges of $A$.
Let $G_{A}$ be the subgraph of $G$ induced by $V(A)$ and take
$\hat{\mc} = (G_{A}, M^{*}_{A}, M_{A})$.
Observe that $\hat{\mc}$ is a valid MC, since $M^{*}_{A}$ and $M_{A}$ are
still a minimum-cost matching and a stable matching, respectively, in
$G_{A}$.
By the choice of $A$, it follows that
$\ratio(\hat{\mc}) \geq \ratio(\mc)$.
\end{proof}

\begin{definition}[Weighted Cycle MC]
A MC $\mc = (G, M^*, M)$ is said to be a \emph{weighted
cycle MC} if $\mc$ is spanned by a single alternating cycle $A$ and
the edge weights in $G$ agree with the distances in the subgraph of $G$
induced by the edges in $E(A)$.
\end{definition}

Our next lemma states that it suffices to bound the \PoA{} in weighted cycle
MCs.

\begin{lemma} \label{lemma:CycleConfigurations}
For every MC $\mc = (G, M^{*}, M)$ on $2 n$ vertices which is spanned by a
single alternating cycle, there exists a weighted cycle MC $\hat{\mc}$ on $2
n$ vertices such that
$\ratio(\hat{\mc}) \geq \ratio(\mc)$.
\end{lemma}
\begin{proof}
Let $A$ be the single alternating cycle spanning $\mc$.
If $\mc$ is not a weighted cycle MC, then $G$ must admit a \emph{shortcut} ---
an edge $(x, y) \in E(G) - E(A)$ satisfying
$w(x, y) < \dist{A}(x, y)$, where $\dist{A}(x, y)$ denotes the distance
between $x$ and $y$ in the (weighted) cycle $A$.
Let $(x, y)$ be a shortcut minimizing $w(x, y)$ and let $z \in V(G) \setminus
\{x, y\}$ be the vertex minimizing $w(x,z) + w(z,y)$.
Observe that $w(x, y)$ must be strictly smaller than $w(x,z) + w(z,y)$ as
$(x, y)$ is a shortcut of $G$ and $G$ does not admit any shorter shortcut.
We argue that the weight of $(x, y)$ can be increased to $w(x, z) + w(z, y)$
without violating the validity of $\mc$ as a MC.
The assertion follows since by repeating this step (finitely many times),
we remove all the shortcuts of $G$.
To that end, note that after increasing $w(x, y)$ to $w(x, z) + w(z, y)$,
$M^{*}$ remains a minimum-cost matching of $G$ (we only increased the weight
of some edge not in $M^{*}$) and $M$ remains a stable matching of $G$ (we only
increased the weight of some edge not in $M$).
So, all we have to show is that $G$ remains metric, which follows from the
choice of $z$.
\end{proof}

\begin{definition}[Weighted Line MC]
We say that a $(2 n)$-vertex metric graph $G$ is a \emph{weighted line graph}
if it can be isometrically embedded into the real line.
As such, it is convenient to identify the vertices of $G$ with the reals
$x_{1} < \cdots < x_{2 n}$ so that $w(x_{i}, x_{j}) = x_{j} - x_{i}$ for
every $1 \leq i < j \leq 2 n$.
In some cases, it will also be convenient to define a weighted line graph by
setting the all differences $x_{i + 1} - x_{i}$ without explicitly specifying
the $x_{i}$s themselves.
A \emph{weighted line MC} $\mc = (G, M^*, M)$ is a MC on $2 n$ vertices
satisfying:
(1) $G$ is a weighted line graph;
(2) $M^{*} = \{ (x_{2 i - 1}, x_{2 i}) \mid 1 \leq i \leq n \}$; and
(3) $M = \{ (x_{2 i}, x_{2 i + 1}) \mid 1 \leq i < n \} \cup \{ (x_1, x_{2 n})
\}$.
Observe that $\mc$ is spanned by a single alternating cycle
$A = (x_1, \dots, x_{2 n}, x_1)$.
\end{definition}

Note that requirement (2) in the definition is not really necessary:
the requirement that $G$ is a weighted line graph already implies that $\{
(x_{2 i - 1}, x_{2 i}) \mid 1 \leq i \leq n \}$ is the unique minimum-cost
matching of $G$ as every other matching $M'$ contains some edge $(x_{i},
x_{j})$ such that $x_{j} - x_{i} > 1$;
it is easy to show that such an edge must belong to an improving alternating
cycle, hence $M'$ cannot be optimal.
Given a $(2 n)$-vertex weighted line graph $G$, we shall subsequently denote
this unique minimum-cost stable matching by $M^{*}(G)$ and the matching
$\{ (x_{2 i}, x_{2 i + 1}) \mid 1 \leq i < n \} \cup \{ (x_1, x_{2 n}) \}$
by $M(G)$.
By definition, $\mc = (G, M^{*}(G), M(G))$ is a valid (weighted line) MC if
and only if $M(G)$ is stable.
Note also that a weighted line MC is a refinement of a weighted cycle MC, with
the additional requirement that the weight of the longest edge in the unique
alternating cycle $A$ equals the total weight of all other edges of $A$.
Building on this fact, the next lemma states that it suffices to consider
weighted line MCs.

\begin{lemma} \label{lemma:LineConfigurations}
For every weighted cycle MC $\mc = (G, M^{*}, M)$ on $2 n$ vertices, there
exists a weighted line MC $\hat{\mc}$ on $2 n$ vertices such that
$\ratio(\hat{\mc}) \geq \ratio(\mc)$.
\end{lemma}
\begin{proof}
Let $A$ be the single alternating cycle spanning $\mc$ and let $e$ be an edge
in $M$ that maximizes $w(e)$.
Let $W_{-e} = \sum_{e' \in E(A) \setminus \{ e \}} w(e')$.
Clearly, $w(e) \leq W_{-e}$, as otherwise, $G$ is not metric.
We argue that if $w(e) < W_{-e}$, then the weight of $e$ can be increased to
$W_{-e}$ without violating the validity of $\mc$ as a MC;
the assertion follows because this step turns $\mc$ into a weighted line MC.
To that end, note that after increasing $w(e)$ to $W_{-e}$, $G$ remains metric
($\mc$ is a weighted cycle MC) and $M^{*}$ remains a minimum-cost matching (we
only increased the weight of some edge not in $M^{*}$).
So, all we have to show is that $M$ remains stable, which follows from the
choice of $e$.
\end{proof}

Once we restrict our attention to weighted line configurations, we can augment
$G$ with new vertices without significantly affecting the ratio of the MC.

\begin{lemma} \label{lemma:AddingPairs}
For every weighted line MC $\mc = (G, M^{*}, M)$ on $2 n$ vertices and for any
$\epsilon > 0$, there exists a weighted line MC $\hat{\mc}$ on $2 (n + 1)$
vertices such that
$\ratio(\hat{\mc}) \geq \ratio(\mc) - \epsilon$.
\end{lemma}
\begin{proof}
Recall that the vertices of $G$ are identified with the reals $x_1 < \ldots <
x_{2 n}$.
Let $\hat{G}$ be the weighted line graph obtained from $G$ by augmenting
$V(G)$ with two new vertices identified with the reals
$y = x_{2 n} + \delta$
and
$y' = y + \delta'$
for some sufficiently small $\delta' > \delta > 0$.
The assertion follows since by taking a sufficiently small $\delta$, we
guarantee that $M(\hat{G})$ is stable in $\hat{G}$, whereas by taking a
sufficiently small $\delta'$, we guarantee that
$c(M(\hat{G})) / c(M^{*}(\hat{G})) \geq \ratio(\mc) - \epsilon$.
\end{proof}

We now turn to present a family of metric graphs referred to as
\emph{Reingold-Tarjan graphs}, acknowledging Reingold and Tarjan's paper
\cite{RT81}, where these graphs were first introduced.
Consider some integer $k > 0$.
The $k^{\text{th}}$ Reingold-Tarjan graph $\grt^{k}$ is a weighted line graph
whose $2^{k}$ vertices are identified with the reals
$x_{1}^{k} < \cdots < x_{2^{k}}^{k}$.
It is defined recursively:
For $k = 1$, we set $x_{2}^{1} - x_{1}^{1} = 1$.
Assume that $\grt^{k}$ is already defined and let $D^{k} = x_{2^{k}}^{k} -
x_{1}^{k}$ be its \emph{diameter}.
Then, $\grt^{k + 1}$ is defined by placing $2$ disjoint instances of $\grt^{k}$
on the real line with an $S^{k + 1}$ \emph{spacing} between them, i.e.,
$x_{2^{k} + 1}^{k + 1} - x_{2^k}^{k + 1} = S^{k + 1}$,
yielding 
$D^{k + 1} = 2 \cdot D^{k} + S^{k + 1}$.
In the current\footnote{
A generalization of the Reingold-Tarjan graphs is presented
in \Sec{}~\ref{section:LowerBoundPoS}, where we use a different value for
$S^{k}$.
} construction, we set $S^{k} = D^{k - 1}$, thus the diameter of
$\grt^{k}$ satisfies
$D^{k} = 3^{k - 1}$.
Refer to \Figure{}~\ref{fig:bad-case metric} for an illustration.

\begin{figure}[htpb]
	\centering
	\fbox{\includegraphics{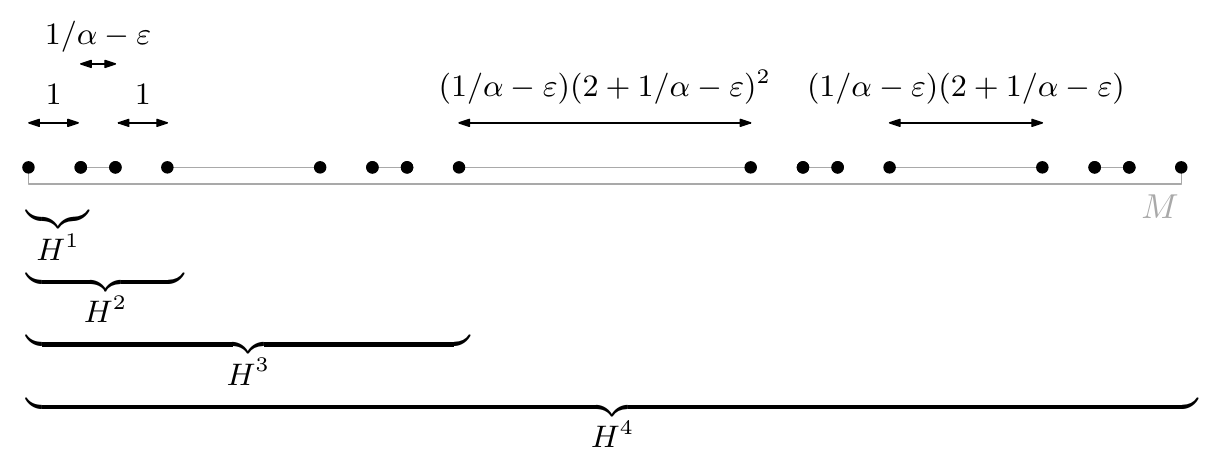}}
\caption{
This extended version of the Reingold-Tarjan graph $\grt^4$ with $2^4$ vertices
has a unique ``expensive'' $\alpha$-stable matching $M$.
Setting the optional parameters $\alpha$ and $\varepsilon$ (that are used in
the proof of the \PoS{} lower bound) to $1$ and $0$, respectively, yields the
original Reingold-Tarjan graph $\grt^4$.
}
\label{fig:bad-case metric}
\end{figure}

Recall that $M^*(\grt^{k})$ matches $x_{2 i - 1}^{k}$ with $x_{2 i}^{k}$ for
every $1 \leq i \leq 2^{k - 1}$;
since all these edges have weight $1$, it follows that $c(M^{*}(\grt^{k})) =
2^{k - 1}$.
Furthermore, we argue by induction on $k$ that the matching
$M(\grt^{k}) = \{ (x_{2 i}^{k}, x_{2 i + 1}^{k}) \mid 1 \leq i < 2^{k} \} \cup \{
(x_{1}^{k}, x_{2^{k}}^{k}) \}$
is stable;
whose cost is
$c(M(\grt^{k})) = D^{k} + (D^{k} - c(M^{*})) = 2 \cdot 3^{k - 1} - 2^{k - 1}$.
Therefore, $\mc_{RT}^{k} = (\grt^{k}, M^{*}(\grt^{K}), M(\grt^{*}))$, referred
to hereafter as the \emph{$k^{\text{th}}$ Reingold-Tarjan MC}, is a valid
weighted line MC with ratio
\[
\ratio(\mc_{RT}^{k}) =
\frac{c(M(\grt^{k}))}{c(M^{*}(\grt^{k}))} =
\frac{2 \cdot 3^{k - 1} - 2^{k - 1}}{2^{k - 1}} =
\Theta \left( (3 / 2)^{k - 1} \right) =
\Theta \left( n^{\log (3 / 2)} \right) \enspace ,
\]
where the last equation follows by setting $2 n = 2^{k}$.
Combined with Lemma~\ref{lemma:AddingPairs}, we immediately conclude that
$\PoA(2 n) = \Omega (n^{\log (3 / 2)})$, establishing the lower bound part of
Theorem~\ref{theorem:PoA}.
The upper bound part of the theorem is established by combining Lemmas
\ref{lemma:SingleAlternatingCycle}, \ref{lemma:CycleConfigurations},
\ref{lemma:LineConfigurations}, and \ref{lemma:AddingPairs} with the following
lemma.

\begin{lemma} \label{lemma:ReingoldTarjan}
The $k^{\text{th}}$ Reingold-Tarjan MC $\mc_{RT}^{k}$ satisfies
$\ratio(\mc_{RT}^{k}) \geq \ratio(\mc)$ for any weighted line MC $\mc$ on
$2^{k}$ vertices.
\end{lemma}
\begin{proof}
By induction on $k$.
The assertion holds trivially for $k = 1$, so assume that it holds for $k$
and consider an arbitrary weighted line MC $\mc = (G, M^{*}(G), M(G))$ on
$2^{k + 1}$ vertices identified with the reals $x_1 < \cdots < x_{2^{k + 1}}$.
Let $L$ and $R$ be the subgraphs of $G$ induced by the vertices $x_1, \dots,
x_{2^{k}}$ and $x_{2^{k} + 1}, \dots, x_{2^{k + 1}}$, respectively.
Let $e = (x_{2^{k}}, x_{2^{k} + 1})$ and let
$D_{L} = x_{2^{k}} - x_{1}$
and
$D_{R} = x_{2^{k + 1}} - x_{2^{k} + 1}$.
We refer to the vertices $x_{1}$ and $x_{2^{k}}$ (respectively, $x_{2^{k} +
1}$ and $x_{2^{k + 1}}$) as the \emph{external} vertices of $L$ (resp.,
$R$) and to the vertices $x_{2}, \dots, x_{2^{k} - 1}$ (resp., $x_{2^{k} + 2},
\dots, x_{2^{k + 1} - 1}$) as the \emph{internal} vertices of $L$ (resp.,
$R$).
Observe that $e \in M(G)$ and since $M(G)$ is a stable matching of $G$, we
must have $x_{2^{k} + 1} - x_{2^{k}} = w(e) \leq \min \{ D_{L}, D_{R} \}$ as
otherwise, at least one of the edges $(x_{1}, x_{2^{k}})$ or $(x_{2^{k} + 1},
x_{2^{k + 1}})$ is unstable. Figure~\ref{Figure:rt-proof} illustrates the various notions.

\begin{figure}
	\centering
	\fbox{\includegraphics{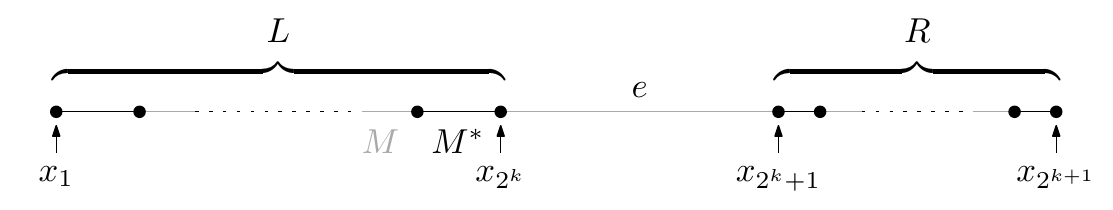}}
	\caption{Any MC $\mc$ on $2^k$ vertices can be transformed into a Reingold-Tarjan $MC$ without decreasing the ratio $\ratio(\mc)$. The black edges are part of the minimum-cost matching $M^*$ while the gray edges belong to the stable matching $M$.}
	\label{Figure:rt-proof}
\end{figure}

We say that a $2^{k}$-vertex weighted line graph is \emph{consistent} with
$\grt^{k}$ if it can be obtained from $\grt^{k}$ by scaling the edge weights.
Fixing the external vertices of $L$ and $R$, we argue that the internal
vertices of $L$ and $R$ can be repositioned so that $L$ and $R$, respectively,
become consistent with $\grt^{k}$ without violating the validity of $\mc$ as
a weighted line MC and without decreasing the ratio $\ratio(\mc)$.
We shall establish this fact for $L$; the proof for $R$ is analogous.
Note first that since $M(\grt^{k})$ is stable in $\grt^{k}$ and since $w(e)
\leq D_{L}$, it follows that by repositioning the internal vertices of $L$ so
that $L$ becomes consistent with $\grt^{K}$, we do not violate the stability of
$M(G)$.
Second, by the inductive hypothesis, repositioning the internal vertices of
$L$ so that $L$ becomes consistent with $\grt^{K}$ maximizes
$c(M(L)) / c(M^{*}(L))$,
thus $\ratio(\mc)$ cannot decrease after this repositioning step, which
establishes the argument.
So, assume hereafter that both $L$ and $R$ are consistent with $\grt^{k}$.

Assume without loss of generality that $D_{L} \geq D_{R}$, so $w(e) = x_{2^{k}
+ 1} - x_{2^{k}}$ is at most $D_{R}$.
In fact, since $R$ is consistent with $\grt^{k}$, it follows
that we can increase the difference $x_{2^{k} + 1} - x_{2^{k}}$ until it is
equal to $D_{R}$, keeping the difference $x_{i + 1} - x_{i}$ unchanged for all
other $i$s, without violating the validity of $\mc$ as a weighted line MC and
without decreasing the ratio $\ratio(\mc)$.
So, assume hereafter that $D_{L} \geq w(e) = D_{R}$.
Now, we argue that we can scale down the differences $x_{i + 1} - x_{i}$ for
every $1 \leq i < 2^{k}$, keeping $x_{i + 1} - x_{i}$ unchanged for all other
$i$s, until we obtain $D_{L} = w(e) = D_{R}$, without decreasing the ratio
$\ratio(\mc)$.
This completes the proof since $D_{L} = w(e) = D_{R}$ implies that $G =
\grt^{k + 1}$.

Let
$\ell = c(M(L)) - D_{L}$,
$\ell^{*} = c(M^{*}(L))$,
$r = c(M(R)) - D_{R}$, and
$r^{*} = c(M^{*}(R))$;
notice that
$\ell + \ell^{*} = D_{L}$ and
$r + r^{*} = D_{R}$.
Since $w(e) = D_{R}$, we can express $\ratio(\mc)$ as
\[
\ratio(\mc) =
\frac{c(M(G))}{c(M^{*}(G))} =
\frac{2 \ell + \ell^{*} + 2 (r + r^{*}) + 2 r + r^{*}}{\ell^{*} + r^{*}} =
\frac{2 \ell + \ell^{*} + 4 r + 3 r^{*}}{\ell^{*} + r^{*}} \enspace .
\]
Recalling that $D_{L} \geq D_{R}$, we express $D_{L}$ as
$D_{L} = (1 + \lambda) D_{R}$ for some $\lambda \geq 0$, and so
$\ell = (1 + \lambda) r$ and
$\ell^{*} = (1 + \lambda) r^{*}$.
Thus,
\[
\ratio(\mc) =
\frac{2 (1 + \lambda) r + (1 + \lambda) r^{*} + 4 r + 3 r^{*}}{(1 + \lambda)
r^{*} + r^{*}} =
\frac{(6 + 2 \lambda) r + (4 + \lambda) r^{*}}{(2 + \lambda) r^{*}} \enspace .
\]
Assuming that the edge weights in $G$ (as a whole) are scaled so that $R =
\grt^{K}$ (rather than merely being consistent with $\grt^{k}$), and recalling
the properties of $\mc_{RT}^{k}$, we get
\[
\ratio(\mc) =
\frac{(6 + 2 \lambda) (3^{k - 1} - 2^{k - 1}) + (4 + \lambda) 2^{k - 1}}{(2 +
\lambda) 2^{k - 1}} =
\left( \frac{6 + 2 \lambda}{2 + \lambda} \right) \cdot (3 / 2)^{k - 1} - 1 \enspace .
\]
The lemma follows since the function
$f(\lambda) = \frac{6 + 2 \lambda}{2 + \lambda}$ is monotonically decreasing
for $\lambda \geq 0$.
\end{proof}

\section{\boldmath$\alpha$-Price of Stability}
\label{section:PoS}

The upper bound established in \Sec{}~\ref{section:PoA} for the
\PoA{} clearly holds for the \PoS{} too;
the matching lower bound can be adapted to the \PoS{} by slightly modifying
the Reingold-Tarjan graphs so that they admit a unique stable matching
\LongVersion 
(see \Sec{}~\ref{section:LowerBoundPoS}),
\LongVersionEnd 
\ShortVersion 
(see \Appendix~\ref{appendix:LowerBoundPoS}),
\ShortVersionEnd 
implying that
$\PoS(2 n) = \Theta (n^{\log (3 / 2)})$.
So, the \PoS{} does not provide much of an improvement over the \PoA{}.
Consequently, we turn to analyze the \PoS{} with respect to relaxed stable
matchings, establishing the following theorem.

\begin{theorem} \label{theorem:PoS}
The $\alpha$-\PoS{} of minimum-cost perfect matchings in metric graphs with $2
n$ vertices is $\Theta (n^{\log (1 + 1 / (2 \alpha))})$.
In particular, taking $\alpha =\BigO (\log n)$ guarantees a constant \PoS{}.
\end{theorem} 

The upper bound promised by Theorem~\ref{theorem:PoS} is constructive, relying
on an efficient greedy algorithm presented in
\Sec{}~\ref{section:GreedyAlgorithm}.
\Sec{}~\ref{section:AnalysisSpecialCase} provides a simplified version of
the analysis of that greedy algorithm that holds only for the case  of $\alpha
= \BigO (\log n)$.
A more involved analysis that covers the general case is given in
\Sec{}~\ref{section:AnalysisGeneralCase}.
The matching $\Omega (n^{\log (1 + 1 / (2 \alpha))})$ lower bound on
$\PoS_{\alpha}(2 n)$ is established via a generalization of the
Reingold-Tarjan graphs
\ShortVersion 
and is deferred to \Appendix{}~\ref{appendix:LowerBoundPoS}.
\ShortVersionEnd 
\LongVersion 
in \Sec{}~\ref{section:LowerBoundPoS}.
\LongVersionEnd 
%

\subsection{Greedy Algorithm for \boldmath$\alpha$-Stable Matchings}
\label{section:GreedyAlgorithm}

The following algorithm called \Greedy transforms a minimum-cost matching $M^*$ in a metric graph into an
$\alpha$-stable matching $M$.

Start with the minimum-cost matching $M \gets M^*$ and iterate over all edges of $G$ by non-decreasing order of weights. If the edge $(u,v)$ currently considered is unstable with respect to the current matching $M$, set $M \gets M \cup \{ (u,v),\ (M(u), M(v)) \}  - \{ (u, M(u)),\ (v, M(v)) \}$ (this operation is called a \emph{flip} of the edge $(u,v)$) and continue with the next edge. After having iterated over all edges, return $M$.

We assume that edge weight ties are resolved in an arbitrary but consistent
manner.
In the following, we denote by $M_i$ the matching calculated by the above
algorithm at the end of iteration $i$.
Moreover, $M_0 = M^*$ is the initial minimum-cost matching and $M_G$ the final
matching returned by \Greedy. The following lemma
\ShortVersion 
(proof deferred to \Appendix{}~\ref{appendix:AdditionalProofsSectionPoS})
\ShortVersionEnd 
shows that the algorithm terminates.

\begin{lemma} \label{lemma:UnstableEdgeCreation}
For any unstable edge $b$ created by the flip of an edge $e$, we have $w(b) >
w(e)$.
\end{lemma}
\newcommand{\ProofLemmaUnstableEdgeCreation}{
We consider the edge $e = (u,v)$ being flipped and we denote by $e' = (M(u),
M(v))$ the second new edge joining $M$ as a result of the flip.
The two edges that are removed by the flip are denoted by $f$ and $g$.
See \Figure{}~\ref{fig:flip-creates-unstable-edges} for an illustration of the situation.
	
	\begin{figure}[htpb]
		\centering
		\fbox{\includegraphics{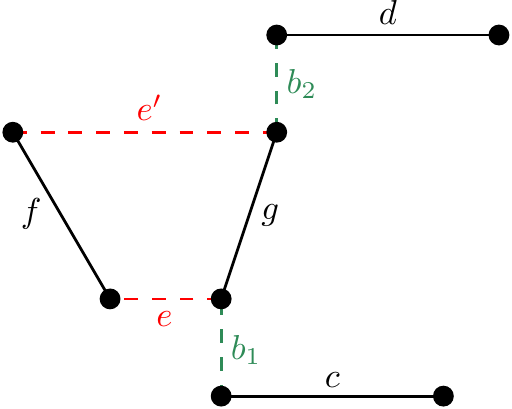}}
		\caption{This figure illustrates the two different cases of Lemma~\ref{lemma:UnstableEdgeCreation}.}
		\label{fig:flip-creates-unstable-edges}
	\end{figure}
	
	When an edge $e$ is flipped, there are essentially two different cases for an unstable edge to be created. Either the unstable edge contains one vertex of $e$ or one vertex of $e'$. No other vertices are involved in the flip and thus every \emph{new} unstable edge has to contain at least one of the four vertices. We assume without loss of generality that a vertex of the edge $g$ is incident to the unstable edge created by the flip.
	
	Let us first consider the case where a vertex of $e$ is incident to the new unstable edge.
	This case is denoted as the edge $b_1$ in \Figure{}~\ref{fig:flip-creates-unstable-edges}. We assume that $b_1$ is stable before the flip and unstable thereafter. For $b_1$ to be unstable after the flip, we must have $\alpha \cdot w(b_1) < w(e)$ and $\alpha \cdot w(b_1) < w(c)$. But as $e$ is unstable before the flip, we have $\alpha \cdot w(e) < w(g)$ and thus we get $\alpha \cdot w(b_1) < w(e) < w(g)/\alpha \leq w(g)$.
	This means that $b_1$ was already unstable before the flip, which is a contradiction to the assumption. Hence, no vertex of $e$ can be part of the new unstable edge.
	
Let us now consider the case, where a vertex from $e'$ is part of the
new unstable edge ($b_2$ in \Figure{}~\ref{fig:flip-creates-unstable-edges}).
Since $b_2$ is stable before the flip and unstable after it, we must have
$w(g) \leq \alpha \cdot w(b_2) < w(e')$.
But as $e$ is unstable before the flip, we have $\alpha \cdot w(e) < w(g)$,
and thus we get
$w(e) < w(g) / \alpha \leq w(b_2)$
which completes the proof.
} 
\LongVersion 
\begin{proof}
\ProofLemmaUnstableEdgeCreation{}
\end{proof}
\LongVersionEnd 

Corollary~\ref{corollary:unstable edges are heavier} follows by induction on
$i$.
\ShortVersion 
Lemma~\ref{lemma:OutputsValidMatching} (proof deferred
to \Appendix{}~\ref{appendix:AdditionalProofsSectionPoS}) then follows by a
straightforward analysis of the algorithm's run-time.
\ShortVersionEnd 

\begin{corollary} \label{corollary:unstable edges are heavier}
Let $e_i$ be the edge considered in iteration $i$.
Then $w(e_i) < w(b)$ for any unstable edge $b$ in $M_{i}$.
\end{corollary}

\begin{lemma} \label{lemma:OutputsValidMatching}
\Greedy transforms a minimum-cost matching into a valid
$\alpha$-stable matching in time $\BigO(n^2 \log n)$.
\end{lemma}
\newcommand{\ProofLemmaOutputsValidMatching}{
The running time of the algorithm is dominated by the step of sorting
the edges in $G$ according to their weight. This takes $\BigO(n^2 \log
n)$ steps. The second phase --- the actual algorithm --- runs in $\BigO(n^2)$
steps since it iterates once over all edges in $V \times V$ and each
iteration takes $\BigO (1)$ time.
	
The correctness of the algorithm is established by
Corollary~\ref{corollary:unstable edges are heavier} since it states that in
the last iteration, all unstable edges have strictly larger weight than the
edge currently considered. Since this edge is already the one with the largest
weight, there cannot be any unstable edges in the final matching $M_G$.
} 
\LongVersion 
\begin{proof}
\ProofLemmaOutputsValidMatching{}
\end{proof}
\LongVersionEnd 

\subsection{Cost Analysis}
\label{section:AnalysisSpecialCase}

In this section, we want to bound the cost of the $\alpha$-stable matching
returned by \Greedy relative to the cost of $M^{*}$.
To this end, we will transcribe the changes that \Greedy 
performs on the minimum-cost matching through a collection of logical
rooted trees, referred to as the \emph{flip forest}, and assign weights
to the nodes of the trees in this forest that will then allow us to derive an
upper bound on the cost of the $\alpha$-stable matching returned by the
algorithm. 

Since this section makes heavy use of rooted binary trees and their
properties, we require a few definitions.
In a \emph{full binary tree}, each inner node has exactly two children.
The \emph{depth} $d(v)$ of a node $v$ in a tree $T$ is the length of the
unique path from the root of $T$ to $v$ and the \emph{height} $h(T)$ of a tree
$T$ is defined as the maximal depth of any node in $T$.
The \emph{height} $h(v)$ of a node $v$ of $T$ is defined to be the height of its subtree.
The \emph{leaf set} $\mathcal L(T)$ or $\mathcal L(F)$ of a tree $T$ or a collection $F$ of trees
is the set of all leaves in $T$ or $F$, respectively.
The \emph{leaf set} $\mathcal L(v)$ of a node $v$ in a tree is $\mathcal L(T_v)$ where $T_v$ is the subtree rooted at $v$.
Finally, two nodes with the same parent are called \emph{sibling nodes}.

We begin with Lemma~\ref{lemma:EdgeUnstableOnlyOnce}
\ShortVersion 
(proof deferred to \Appendix{}~\ref{appendix:AdditionalProofsSectionPoS})
\ShortVersionEnd 
stating an important property of the edges that are flipped by \Greedy.

\begin{lemma} \label{lemma:EdgeUnstableOnlyOnce}
If an edge $e$ is flipped in iteration $i$, then $e \in M_j$ for all $j \geq
i$ and in particular $e \in M_G$.
\end{lemma}
\newcommand{\ProofLemmaEdgeUnstableOnlyOnce}{
Let us assume for the sake of contradiction that $e = (u,v)$ was flipped in iteration $i$ of the algorithm and further that $(u, v) \notin M_j$ for some $j > i$. According to the algorithm, we have $(u,v) \in M_i$. Since $(u,v) \notin M_j$, there has to exist an iteration $k$ with $i < k \leq j$ where $(u,v)$ is removed from $M_{k-1}$ such that $(u,v) \notin M_k$. For this to happen, either edge $(u, u')$ or $(v, v')$ for some vertex $u'$ or $v'$ must be flipped in iteration $k$ because it was unstable with respect to $M_{k-1}$. Without loss of generality, we assume that $(u, u')$ is unstable with respect to $M_{k-1}$ and flipped in iteration $k > i$ and we have
\begin{MathMaybe}
	w(u, u') ~ \leq ~ \alpha \cdot w(u, u') < w(u, v)  \enspace .
\end{MathMaybe}
But this means that \Greedy would have considered the edge $(u,u')$ before considering the edge $(u,v)$, a contradiction to the assumption.
} 
\LongVersion 
\begin{proof}
\ProofLemmaEdgeUnstableOnlyOnce{}
\end{proof}
\LongVersionEnd 

Consider an iteration of \Greedy where edge $(u,v)$ is flipped
because it was unstable at the beginning of the iteration.
Then the two edges $(u, M(u))$ and $(v, M(v))$ are replaced by $(u, v)$ and
$(M(u), M(v))$.
Since, according to Lemma~\ref{lemma:EdgeUnstableOnlyOnce}, the edge $(u,
v)$ is selected irrevocably, the edges $(u, M(u))$ and $(v, M(v))$ can never
be part of $M$ again.
The only edge, of the four edges involved, that may be changed again, is the
edge $(M(u), M(v))$.
Thus, we refer to $(M(u), M(v))$ as an \emph{active} edge.
We also refer to all edges in $M_{0}$ as active.
Using the notion of active edges, we shall now model the changes that \Greedy
applies to the matching during its execution through a logical helper structure
called the flip forest.

\begin{definition}[Flip Forest]
The \emph{flip forest} $F = (U,K)$ for a certain execution of \Greedy is a
collection of rooted trees with node set $U$ and link set $K$.
It contains a node $u_e \in U$ corresponding to each edge $e \in V \times V$
that has been active at some stage during the execution.
This correspondence is denoted by $u_e \cor e$.
For each flip of an edge $(u,v)$ in $G$, resulting in the removal of the edges
$(u, M(u))$ and $(v, M(v))$ from $M$, $K$ contains a link connecting the node
$y \cor (u, M(u))$ to its parent $x \cor (M(u), M(v))$ and a link connecting
the node $z \cor (v, M(v))$ to its parent $x \cor (M(u), M(v))$.
(Observe that by definition, all three edges $(u, M(u))$, $(v, M(v))$, and
$(M(u), M(v))$ are active.)
Refer to \Figure{}~\ref{fig:tree segment} for an illustration.
\end{definition}

\begin{figure}
\centering
\fbox{\includegraphics{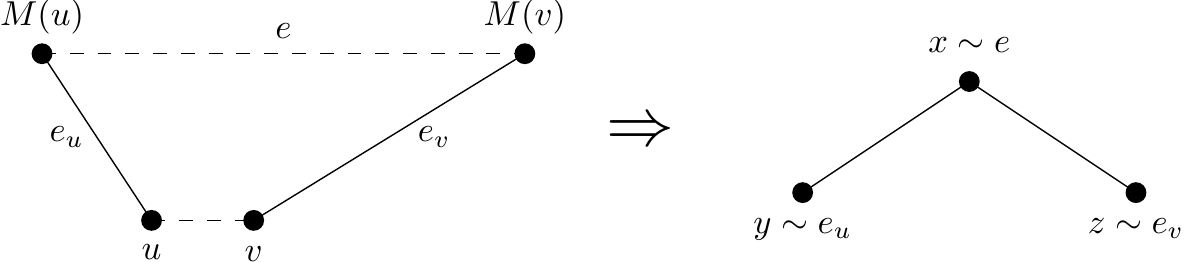}}
\caption{
The left side shows a matching configuration with an unstable edge $(u,v)$,
which will be flipped by \Greedy.
This flip is then represented by the flip tree segment on the right, which
depicts the replacement of the two active edges $(u,M(u)) \cor y$ and $(v,
M(v)) \cor z$ by the active edge $(M(u), M(v)) \cor x$.}
\label{fig:tree segment}
\end{figure}

To avoid confusion between the basic elements of $G$ and the basic elements of
$F$, we refer to the former as vertices/edges and to the latter as
nodes/links.

The definition of a flip forest ensures that for each flip of the
algorithm, we obtain a binary \emph{flip tree segment} as depicted by
\Figure{}~\ref{fig:tree segment}.
When we transcribe each flip operation of the complete execution of 
\Greedy into a flip tree segment as explained above, we end up with
a collection of full binary trees --- a forest as depicted in
\Figure{}~\ref{fig:flip forest}.
This is because the parent node of a tree segment may appear as a child node
of the tree segment corresponding to a later iteration of the algorithm since
its corresponding edge is still active and therefore may participate in
another flip.
Each such tree is called a \emph{flip tree} hereafter.
Observe that all leaves in the flip forest correspond to edges in the
minimum-cost matching $M_0 = M^*$.

\begin{figure}
\centering
\fbox{\includegraphics{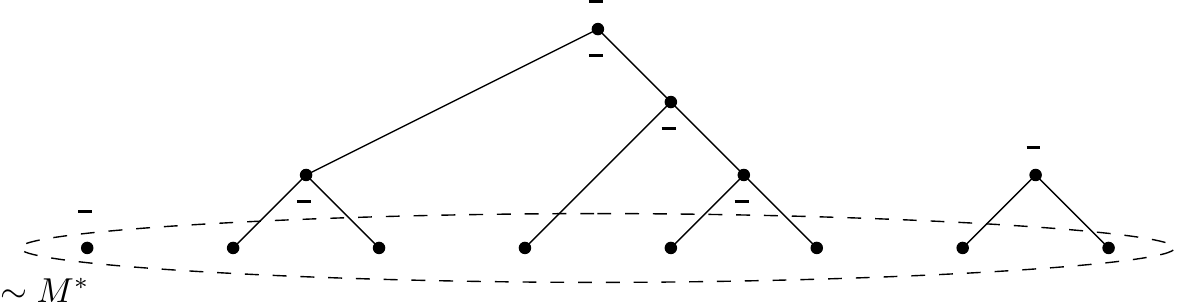}}
\caption{All leaves and isolated nodes of the flip forest $F$ correspond to edges in the minimal-cost matching $M^*$. Each inner node corresponds to the active edge that resulted from the respective flip. Note that the edge that got flipped and is therefore irrevocably selected into $M_G$ has no corresponding node in $F$. For the purpose of illustration, we can associate such an edge with the respective node as indicated by a line \emph{below} the respective inner node. These edges constitute the matching $M_G$ together with the edges corresponding to isolated vertices and roots, indicated by a line \emph{above} the node.
	}
  \label{fig:flip forest}
\end{figure}

We now define a function $\wb : U \mapsto \mathbb R$ that maps a \emph{virtual
weight} to each node in the flip forest $F$ as follows.
For each leaf $\ell$ of a flip tree in $F$, we set $\wb(\ell) \coloneqq w(e)$,
where $\ell \cor e$ and we recall that an edge corresponding to a leaf node in
$F$ is part of $M^*$.
The function $\wb$ is extended to an inner node $x$ of a flip tree with child
nodes $y$ and $z$ by the recursion 
\begin{equation}
	\wb(x) \coloneqq \wb(y) + \wb(z) + (1/\alpha) \cdot \min \{\wb(y), \wb(z) \} \enspace . \label{eq:wb definition} 
\end{equation}
For the ease of argumentation, we call the child with smaller (respectively,
larger) value of $\wb$ as well as the link leading to its parent \emph{light}
(resp., \emph{heavy}).
We denote the light child of a node $x$ as $\lc{x}$ and the heavy child as
$\hc{x}$.
Then we can rewrite the recursion from \Equation{}~\ref{eq:wb definition} as 
\begin{MathMaybe}
	\wb(x) \coloneqq \wb(\hc{x}) + (1+1/\alpha) \cdot \wb(\lc{x}) \enspace .
\end{MathMaybe}

\begin{lemma} \label{lemma:FlipEdgeWeightBound}
Let $x$ be a node in $F$ and $e$ an edge in $G$ with $x \cor e$.
Then $w(e) \leq \wb(x)$.
\end{lemma}
\begin{proof}
We prove the statement by induction over the height of $x$ in its flip tree.
The assertion holds for every leaf $x \cor e$ in the flip forest as $\wb(x) =
w(e)$ by definition.
Assume that the statement holds for the two children $\lc{x}$ and $\hc{x}$ of
a node $x$ that represents a flip of the edge $(u, v)$.
Then $x \cor (M(u), M(v)) = e$ and we assume without loss of generality that
$\hc{x} \cor (u, M(u)) = e_u$ and
$\lc{x} \cor (v, M(v)) = e_v$.
Thus,
$w(e_u) \leq \wb(\hc x)$ and
$w(e_v) \leq \wb(\lc x)$.
This flip tree segment represents the replacement of the edges $e_u$ and $e_v$
by $e$ and $(u,v)$, which happened because the edge $(u, v)$ was unstable with
respect to $M$, that is, $\alpha \cdot w(u,v) < \min \{ w(e_v), w(e_u) \}$.
Since $G$ is metric, we can bound $w(e)$ as
\ShortVersion
\begin{equation*}
	w(e) ~ \leq ~ w(e_u) + w(e_v) + w(u,v) ~ < ~ w(e_u) + (1+1/\alpha)
\cdot w(e_v) ~ \leq ~ \wb(\hc{x}) + (1+1/\alpha) \cdot \wb(\lc{x}) ~ = ~
\wb(x) \, ,
\end{equation*}
where the last inequality holds by the inductive hypothesis.
\ShortVersionEnd
\LongVersion
\begin{align*}
	w(e) ~ &\leq ~ w(e_u) + w(e_v) + w(u,v) \\
	&< ~ w(e_u) + (1+1/\alpha) \cdot w(e_v) \\
	&\leq ~ \wb(\hc{x}) + (1+1/\alpha) \cdot \wb(\lc{x}) \tag*{\hspace{-2cm}(inductive hypothesis)} \\
	&= ~ \wb(x) \enspace . \tag*{\qedhere}
\end{align*}
\LongVersionEnd
\end{proof}

\begin{definition}[Light Depth]
The \emph{light depth} $\lambda(x)$ of a node $x$ in a flip forest $F$ is the
number of light links on the direct path from $x$ to the root of the flip tree
containing $x$.
\end{definition}

\ShortVersion
Lemma~\ref{Lemma:NodeWeightBound} (proof deferred to \Appendix{}~\ref{appendix:AdditionalProofsSectionPoS}) relates the virtual weight of an inner node of a flip tree to the virtual weights of the leaves of its subtree.
\ShortVersionEnd

\begin{lemma} \label{Lemma:NodeWeightBound}
Every node $x$ in a flip tree satisfies
\begin{MathMaybe}
\wb(x) ~ = ~ \sum_{\ell \in \mathcal L(x)} (1+1/\alpha)^{\gl(\ell)-\gl(x)}
\cdot \wb(\ell) \enspace .
\end{MathMaybe}
\end{lemma}
\newcommand\ProofLemmaNodeWeightBound{
We prove the statement by induction over the height of $x$ in its flip tree.
The statement holds for a leaf node $x$ since then we have $\mathcal L(x) = \{
x \}$ and $\lambda(x) - \lambda(x) = 0$.
Assume that the statement holds for both children $\hc{x}$ and $\lc{x}$ of a
node $x$.
By definition, we have
\begin{align*}
		\wb(x) ~ &= ~ \wb(\hc{x}) + (1+1/\alpha) \cdot \wb(\lc{x}) \\ 
		&= ~ \sum_{\ell \in \mathcal L(\hc{x})} (1+1/\alpha)^{\gl(\ell)-\gl(\hc{x})} \cdot \wb(\ell) + (1+1/\alpha) \cdot \sum_{\ell \in \mathcal L(\lc{x})} (1+1/\alpha)^{\gl(\ell)-\gl(\lc{x})} \cdot \wb(\ell) \\
		&= ~ \sum_{\ell \in \mathcal L(\hc{x})} (1+1/\alpha)^{\gl(\ell)-\gl(x)} \cdot \wb(\ell) + (1+1/\alpha) \cdot \sum_{\ell \in \mathcal L(\lc{x})} (1+1/\alpha)^{\gl(\ell)-\gl(x)-1} \cdot \wb(\ell) \\
		&= ~ \sum_{\ell \in \mathcal L(\hc{x})} (1+1/\alpha)^{\gl(\ell)-\gl(x)} \cdot \wb(\ell) + \sum_{\ell \in \mathcal L(\lc{x})} (1+1/\alpha)^{\gl(\ell)-\gl(x)} \cdot \wb(\ell) \\
		&= ~ \sum_{\ell \in \mathcal L(x)}
(1+1/\alpha)^{\gl(\ell)-\gl(x)} \cdot \wb(\ell) \enspace ,
\end{align*}
where we used $\lambda(\lc x) = \lambda(x)+1$ and $\lambda(\hc x) =
\lambda(x)$.
}

\LongVersion
\begin{proof}
\ProofLemmaNodeWeightBound
\end{proof}
\LongVersionEnd

\noindent
Corollary~\ref{corollary:RootWeight} is immediate, since $\gl(r_T) = 0$ for
the root $r_T$ of a flip tree $T$.

\begin{corollary} \label{corollary:RootWeight}
The root $r_T$ of a flip tree $T$ satisfies
\begin{MathMaybe}
\wb(r_T) ~ = ~ \sum_{\ell \in \mathcal L(r_T)} (1+1/\alpha)^{\gl(\ell)} \cdot
\wb(\ell) \enspace .
\end{MathMaybe}
\end{corollary}

The following observation stems from the fact that in each segment, the
$\wb$-value of the parent is at least $(2 + 1 / \alpha)$ times that of the
light child (equality holds when both children have the same $\wb$-value).

\begin{observation} \label{observation:LeafWeightUpperBound}
For any flip tree $T$ with root $r_T$ and any leaf $\ell$ of $T$, we have
\begin{MathMaybe}
\wb(r_T) ~ \geq ~ (2+1/\alpha)^{\lambda(\ell)} \cdot \wb(\ell) \enspace .
\end{MathMaybe}
\end{observation}

We now turn to bound $\wb(r_T)$ for all trees $T \in F$ with respect to the sum of
the weights of the edges that correspond to the leaves of $T$.
Since all these edges are part of $M^*$ by construction of $F$, this will
allow us to bound the cost of $M_G$ with respect to $M^*$. 

\begin{lemma} \label{lemma:RootWeightBound}
The virtual weights in a flip tree $T$ satisfy
\begin{MathMaybe}
\wb(r_T) ~ = ~
\BigO \Big( (1 + 1 / \alpha)^{\log n} \sum_{\ell \in \mathcal L(r_T)}
\wb(\ell) \Big) \enspace .
\end{MathMaybe}
\end{lemma}
\begin{proof}
Corollary~\ref{corollary:RootWeight} implies that
$\wb(r_T) = \sum_{\ell \in \mathcal L(r_T)} (1+1/\alpha)^{\gl(\ell)} \cdot \wb(\ell)$.
We group the leaves according to their light depth, where $\mathcal L_j$
denotes the set of leaves $\ell$ of $T$ with $\gl(\ell) = j$.
The equation for $\wb(r_T)$ can now be rewritten as
$\wb(r_T) = \sum_{j=0}^n \ \sum_{\ell \in \mathcal L_j} (1+1/\alpha)^{j} \cdot
\wb(\ell)$.
Let $\Wb^{>} = \sum_{j = j'}^n \sum_{\ell \in \mathcal L_j}
(1+1/\alpha)^{j} \cdot \wb(\ell)$ for some $j'$ that will soon be determined.
We apply Observation~\ref{observation:LeafWeightUpperBound} and the fact that
there are at most $n$ leaves in $T$ altogether to conclude
\ShortVersion
\begin{align*}
	\Wb^{>} ~ &= ~ \sum_{j=j'}^{n} \ \sum_{\ell \in \mathcal L_j} (1+1/\alpha)^{j} \cdot \wb(\ell) ~ \leq ~ \sum_{j=j'}^{n} \ \sum_{\ell \in \mathcal L_j} \left(\frac{1+1/\alpha}{2+1/\alpha} \right)^{\!\!j} \cdot \wb(r_T) \\
	&= ~ \sum_{j=j'}^{n} |\mathcal L_j| \left(\frac{1+1/\alpha}{2+1/\alpha}\right)^{\!\!j} \cdot \wb(r_T) ~ \leq n ~ \cdot \left(\frac{1+1/\alpha}{2+1/\alpha}\right)^{\!\!j'}
\cdot \wb(r_T) \, .
\end{align*}
\ShortVersionEnd
\LongVersion
\begin{align*}
	\Wb^{>} ~ &= ~ \sum_{j=j'}^{n} \ \sum_{\ell \in \mathcal L_j} (1+1/\alpha)^{j} \cdot \wb(\ell) \\
	&\leq ~ \sum_{j=j'}^{n} \ \sum_{\ell \in \mathcal L_j} \left(\frac{1+1/\alpha}{2+1/\alpha} \right)^{\!\!j} \cdot \wb(r_T) \\
	&= ~ \sum_{j=j'}^{n} |\mathcal L_j| \left(\frac{1+1/\alpha}{2+1/\alpha}\right)^{\!\!j} \cdot \wb(r_T) \\
	&\leq ~ n \cdot \left(\frac{1+1/\alpha}{2+1/\alpha}\right)^{\!\!j'}
\cdot \wb(r_T) \enspace .
\end{align*}
\LongVersionEnd
Choosing
$j' = \log_{\frac{2 + 1 / \alpha}{1 + 1 /\alpha}} (2 n) = \BigO (\log n)$
yields
$\Wb^{>} \leq \wb(r_T) / 2$.
This means that the leaves with light depth at most $c \log n$ for some
constant $c$ contribute at least half of $\wb(r_T)$ and thus it suffices to
consider only those leaves in order to bound $\wb(r_T)$:
\begin{align*}
\wb(r_T) & ~ \leq ~ 2 \cdot \sum_{j=0}^{c \log n} \ \sum_{\ell \in \mathcal L_j}
(1 + 1 / \alpha)^{j} \cdot \wb(\ell) ~ \leq ~ 2 \cdot (1 + 1 / \alpha)^{c \log n} \sum_{\ell \in \mathcal L(r_T)}
\wb(\ell) \enspace . \tag*{\qedhere}
\end{align*}
\end{proof}

At this stage, we would like to relate the virtual weight $\wb(r_T)$ of the
roots $r_T$ in $F$ to the cost of the stable matching $M_G$ returned by 
\Greedy. 
To that end, we observe that $M_G$ consists of the edges corresponding to the
roots in $F$ and to the edges that have been flipped along the course of the
execution;
let $D$ denote the set of the latter edges.

Consider the flip of edge $(u, v)$, resulting in the insertion of edge $(M(u),
M(v)) \cor x$ to $M$ and the removal of edges $(u, M(u)) \cor x_{L}$ and $(v,
M(v)) \cor x_{H}$ from $M$.
Since $\wb(x) = \wb(x_{H}) + (1 + 1 / \alpha) \wb(x_{L})$, we have
$\wb(x) - (\wb(x_{L}) + \wb(x_{H})) = \wb(x_{L}) / \alpha$.
Lemma~\ref{lemma:FlipEdgeWeightBound} then implies that
$\wb(x) - (\wb(x_{L}) + \wb(x_{H})) \geq w(u, M(u)) / \alpha$, and since edge
$(u, v)$ was flipped, we have $\wb(x) - (\wb(x_{L}) + \wb(x_{H})) \geq w(u,
v)$.
Therefore,
\[
\sum_{e \in D} w(e)
~ \leq ~
\sum_{\text{internal } x \in U} \left( \wb(x) - (\wb(x_{L}) + \wb(x_{H})) \right)
~ = ~
\sum_{\text{flip trees T}} \Big( \wb(r_T) - \sum_{\ell \in \mathcal{L}(T)}
\wb(\ell) \Big) \enspace ,
\]
where the second equation holds by a telescoping argument.
Corollary~\ref{corollary:flip tree cost} follows since $c(M^{*}) = \sum_{\ell \in \mathcal L(F)}
\wb(\ell)$.

\begin{corollary} \label{corollary:flip tree cost}
The matching $M_G$ returned by \Greedy satisfies
$c(M_{G}) \leq 2 \sum_{\text{\rm flip trees $T$}} \wb(r_T) - c(M^{*})$.
\end{corollary}

We are now ready to establish the following lemma.

\begin{lemma}
	\label{lemma:log-stable is c-minimum}
The cost of the matching $M_G$ returned by \Greedy for $\alpha = \BigO(\log
n)$ is an $\BigO (1)$ approximation of $c(M^{*})$.
\end{lemma}

\begin{proof}
Employing Lemma~\ref{lemma:RootWeightBound} and setting $\alpha = \BigO(\log
n)$, we get $\wb(r_T) = \BigO \big(\! \sum_{\ell \in \mathcal{L}(T)} \wb(\ell)
\big)$.
Corollary~\ref{corollary:flip tree cost} and the fact that $c(M^{*}) =
\sum_{\ell \in \mathcal L(F)} \wb(\ell)$ then imply that
$c(M_{G}) = \BigO (c(M^{*}))$ as desired.
\end{proof}

\subsection{Tight Upper Bound}
\label{section:AnalysisGeneralCase}

Our goal in this section is to show that when \Greedy is invoked
with parameter $\alpha$ for any $\alpha \geq 1$, it returns an $\alpha$-stable
matching $M_{G}$ satisfying
$c(M_{G}) = c(M^{*}) \cdot \BigO (n^{\log (1 + 1 / (2 \alpha))})$.
This is performed by taking a deeper examination of the properties of our
flip trees and their virtual weights.
It will be convenient to ignore the relation of the flip trees to the \Greedy
algorithm at this stage;
in other words, we consider an abstract full binary tree $T$ with a function
$w : \mathcal{L}(T) \rightarrow \mathbb{R}_{\geq 0}$ that assigns non-negative
weights to the leaves of $T$, which then determines the virtual weight
$\wb(x)$ of each node in $T$, following the recursion of
\Equation{}~(\ref{eq:wb definition}).
Note that we allow our tree $T$ to have zero-weight leaves now (this can only
make our analysis more general).

\begin{definition}[Complete Binary Tree]
A full binary tree $T$ is called \emph{complete} if all leaves are at depth
$h(T)$ or $h(T) - 1$.
Given some positive integer $n$ that will typically be the number of leaves
in some tree, let
\[ h(n) = \lceil \log n \rceil \quad \text{ and } \quad k(n) = 2^{h(n)} - n \enspace . \]
Note that $0 \leq k(n) < 2^{h(n)-1}$.
\end{definition}

\begin{definition}[\wb-Balanced Flip Tree]
A full binary tree $T$ is called \emph{\wb-balanced} if for any two sibling
nodes $x,y$ in $T$, we have $\wb(x) = \wb(y)$.
\end{definition}

Consider some full binary tree $T$.
Let $\lw(T)$ denote the sum of the virtual weights of $T$'s leaves, that is,
$\lw(T) = \sum_{\ell \in \mathcal{L}(T)} \wb(\ell)$, and let $\Wb(T) =
\wb(r_T)$ (recall that $r_T$ denotes the root of $T$).
The following observation is established by induction on the depth of the
nodes.

\begin{observation} \label{observation:balanced tree node wb}
For any node $v$ of a \wb-balanced full binary tree $T$, we have $\wb(v) =
(2+1/\alpha)^{-d(v)} \cdot \Wb(T)$.
\end{observation}

\begin{definition}[Effect of a Flip Tree]
The \emph{effect} $\eff(T)$ of a full binary tree $T$ is defined to be
\[
\eff(T) ~ = ~
\begin{cases}
\Wb(T) / \lw(T) & \text{if } \lw(T) > 0 \\
1 & \text{if } \lw(T) = 0 
\end{cases} \enspace .
\]
An $n$-leaf full binary tree $T$ is said to be \emph{effective} if it maximizes
$\eff(T)$, namely, if there does not exist any $n$-leaf full binary tree $T'$
such that $\eff(T') > \eff(T)$.
\end{definition}

Intuitively speaking, if we think of $T$ as a flip tree, then its effect is a
measure for the factor by which the flips represented by $T$ increase the cost
of $M^*$ when applied to it.
But, once again, we do not restrict our attention to flip trees at this stage.
The effect of a full binary tree is essentially determined by its topology
and by the assignment of weights to its leaves.
It is important to point out that by Corollary~\ref{corollary:RootWeight}, the
effect of a full binary tree is not affected by scaling its leaf weights.
Our upper bound is established by showing that the effect of an effective
$n$-leaf full binary tree is $\BigO \left( n^{\log (1 + 1 / (2 \alpha))} \right)$.
We begin by developing a better understanding of the topology of effective
$\wb$-balanced full binary trees.
%

\begin{lemma} \label{Lemma:BalancedTreeIsComplete}
An effective $n$-leaf \wb-balanced full binary tree must be complete.
\end{lemma}
\newcommand{\ProofLemmaBalancedTreeIsComplete}{
Aiming for a contradiction, suppose that $T$ is not complete and scale the
leaf weights in $T$ so that $\Wb(T) = 1$.
Because $T$ is not complete, it must have leaves at depth $d_1$ and at depth
$d_2$, where $d_2 > d_1 + 1$.
The assertion is established by showing that an $n$-leaf full binary tree with
higher effect can be obtained by a small modification to $T$'s topology, in
contradiction to the assumption that $T$ is effective.

Let $y$ be a leaf at depth $d_1$ and $\ell_1$ and $\ell_2$ be two leaves at
depth $d_2 > d_1 + 1$ with parent node $z$.
Since $T$ is \wb-balanced, we can employ Observation~\ref{observation:balanced
tree node wb} to conclude that $\wb(\ell_1) = \wb(\ell_2) = (2 + 1 /
\alpha)^{-d_2}$ and $\wb(y) = (2 + 1 / \alpha)^{-d_1}$.

Now, consider the $\wb$-balanced full binary tree $T'$ obtained from $T$ by
removing $\ell_1$ and $\ell_2$ and adding two new leaves $\ell'_1$ and
$\ell'_2$ as children of $y$ with virtual weight $\wb(\ell'_1) = \wb(\ell'_2)
= (2 + 1 / \alpha)^{-d_1 - 1}$, keeping the virtual weight of all other nodes
unchanged.
By doing so, we turn $z$ --- an internal node in $T$ --- into a leaf (whose
virtual weight remains $\wb(z) = (2 + 1 / \alpha)^{-d_2 + 1}$).
On the other hand, $y$ which is a leaf in $T$, is an internal node in $T'$.
Therefore,
\begin{align*}
\lw(T')
~ = ~ &
\lw(T) + \wb(\ell'_1) + \wb(\ell'_2) + \wb(z) - \wb(\ell_1) - \wb(\ell_2) -
\wb(y) \\
= ~ &
\lw(T) + 2 \cdot (2 + 1 / \alpha)^{-d_1 - 1} + (2 + 1 / \alpha)^{-d_2 + 1}
- 2 \cdot (2 + 1 / \alpha)^{-d_2} - (2 + 1 / \alpha)^{-d_1} \\
= ~ &
\lw(T) + (2 + 1 / \alpha)^{-d_1 - 1} (2 - (2 + 1 / \alpha)) + (2 + 1 /
\alpha)^{-d_2} (2 + 1 / \alpha - 2) \\
= ~ &
\lw(T) + (1 / \alpha) ((2 + 1 / \alpha)^{-d_2} - (2 + 1 / \alpha)^{-d_1 - 1})
\\
< ~ &
\lw(T) \enspace .
\end{align*}
As $\Wb(T') = \Wb(T) = 1$, it follows that $\eff(T') > \eff(T)$, in
contradiction to the effectiveness of $T$.
} 

\begin{proof}
  \ProofLemmaBalancedTreeIsComplete
\end{proof}

Next, we develop a closed-form expression for the effect of complete
$\wb$-balanced full binary trees.
 
\begin{lemma}
	\label{lem:expression for eff of balanced complete tree}
The effect of an $n$-leaf complete $\wb$-balanced full binary tree $T$ is
\[ \eff(T) = \frac{(2 + 1 / \alpha)^{h}}{2^{h} + k / \alpha} \enspace ,\]
where $h = h(n)$ and $k = k(n)$.
\end{lemma}

\begin{proof}
	Again we assume without loss of generality that the weights of the leaves are
	scaled so that $\Wb(T) = 1$.
	By definition, $T$ has $2^h - 2k$ leaves at depth $h$ and $k$ leaves at
	depth $h - 1$.
	Employing Observation~\ref{observation:balanced tree node wb}, we conclude
	\begin{align*}
	\lw(T)
	~ = ~ &
	(2^{h} - 2k) \cdot (2 + 1 / \alpha)^{-h} + k \cdot (2 + 1 / \alpha)^{-(h
	- 1)} \\
	= ~ &
	(2 + 1 / \alpha)^{-h} \cdot (2^{h} - 2k + k \cdot (2 + 1 / \alpha)) \\
	= ~ &
	(2 + 1 / \alpha)^{-h} \cdot (2^{h} + k / \alpha ) \enspace .
	\end{align*}
	Since $\Wb(T) = 1$, we have $\eff(T) = 1 / \lw(T)$ which completes the
proof.
\end{proof}

Note that the expression for the effect of an $n$-leaf complete $\wb$-balanced
full binary tree given by Lemma~\ref{lem:expression for eff of balanced
complete tree} is monotonically increasing with $h$ and monotonically
decreasing with $k$.
We are now ready to show that it is essentially sufficient to consider
complete $\wb$-balanced full binary trees.

\begin{lemma}
	\label{lem:eff-max tree is balanced}
An effective $n$-leaf full binary tree must be $\wb$-balanced.
\end{lemma}

\begin{proof}
We prove the statement by induction on the number of leaves $n$.
The base case of a tree having a single leaf (which is also the root) holds
vacuously;
the base case of a tree having two leaves is trivial.
Assume that the assertion holds for trees with less than $n$ leaves and let
$T$ be an effective $n$-leaf full binary tree.
Let $T_{L}$ and $T_{H}$ be the subtrees rooted at the light and heavy,
respectively, children of $r_{T}$ (break ties arbitrarily).
Let $n_L$ and $n_H$ be the number of leaves in $T_L$ and $T_H$, respectively,
where $n_L + n_H = n$ and $0 < n_L, n_H < n$.
Observe that since
$\eff(T) = \frac{\Wb(T_H) + (1+1/\alpha) \cdot \Wb(T_L)}{\lw(T_H) + \lw(T_L)}$,
both $T_L$ and $T_H$ have to be effective as otherwise, $\eff(T)$ could be
increased;
more precisely, if $T_i \in \{ T_H,\, T_L \}$ is not effective, then one can
increase $\Wb(T_i)$ while keeping $\lw(T_i)$ unchanged, which results in an
increased $\eff(T)$.
Thus, by the inductive hypothesis, we conclude that $T_{L}$ and $T_{H}$ must be
$\wb$-balanced.
Lemma~\ref{Lemma:BalancedTreeIsComplete} then guarantees that
both $T_{L}$ and $T_{H}$ are complete.

Aiming for a contradiction, suppose that $T$ is not $\wb$-balanced, that is
$\Wb(T_{H}) > \Wb(T_{L})$.
Assume without loss of generality that the leaf weights are scaled such that
$\lw(T) = \lw(T_{H}) + \lw(T_{L}) = 1$ and set $\lw(T_{L}) = x$, $\lw(T_{H}) =
1 - x$, for some $0 \leq x \leq 1$.
Let $T$ be the tree minimizing $x$ among all trees satisfying the
aforementioned assumptions.

We argue that $x$ cannot be neither $0$ nor $1$.
Indeed, if $x = 1$, then $\Wb(T_H) = 0$, in contradiction to the assumption
that $\Wb(T_H) > \Wb(T_L)$.
On the other hand, if $x = 0$, then $\Wb(T) = \Wb(T_H)$ and $\lw(T) =
\lw(T_H)$, hence $\eff(T) = \eff(T_{H})$.
But since $T_H$ has $n_{H} < n$ leaves, Lemma~\ref{lem:expression for eff of
balanced complete tree} guarantees that its effect is smaller than that of an
$n$-leaf complete $\wb$-balanced full binary tree, in contradiction to the
assumption that $T$ is effective.

So, we may subsequently assume that $0 < x < 1$.
Employing Lemma~\ref{lem:expression for eff of balanced complete tree}, we can
express $\Wb(T)$ as
\[
\Wb(T)
~ = ~
\Wb(T_H) + (1 + 1 / \alpha) \cdot \Wb(T_L)
~ = ~
\frac{(2 + 1 / \alpha)^{h_H}}{2^{h_H} + k_H / \alpha} \cdot (1 - x) +
(1 + 1 / \alpha) \cdot \frac{(2 + 1 / \alpha)^{h_L}}{2^{h_L} + k_L / \alpha} \cdot x \enspace ,
\]
where $h_{H} = h(n_H)$, $k_H = k(n_H)$, $h_L = h(n_L)$, and $k_L = k(n_L)$.
Using this expression, we can formulate $\eff(T)$ as a function $f = f(x)$,
setting
\begin{equation}
f(x)
~ = ~
\frac{(2 + 1 / \alpha)^{h_H}}{2^{h_H} + k_H / \alpha} \cdot (1 - x) +
(1 + 1 / \alpha) \cdot \frac{(2 + 1 / \alpha)^{h_L}}{2^{h_L} + k_L / \alpha} \cdot x \enspace .
\label{eq:f}
\end{equation}
The crucial observation now is that $f(x)$ is linear in $x$, thus $\frac{\mathrm
d f}{\mathrm d x}(x)$ is independent of $x$. Moreover, since $T$ is not
\wb-balanced, it follows that $f$ is well defined --- that is,
\Equation{}~(\ref{eq:f}) remains valid --- in a neighborhood of $x =
\lw(T_L)$.
Therefore, if,
$\frac{\mathrm d f}{\mathrm d x}(x) > 0$, then $f(x) = \eff(T)$ can be increased
by increasing $x$ (shifting weight from the leaves of $T_{H}$ to the leaves of
$T_{L}$), in contradiction to the effectiveness of $T$. On the other hand, if
$\frac{\mathrm d f}{\mathrm d x}(x) \leq 0$, then we can decrease $x$ (shifting
weight from the leaves of $T_{L}$ to the leaves of $T_{H}$) without decreasing
$f(x) = \eff(T)$, contradicting the assumption that $x$ is minimum.
The assertion follows.
\end{proof}

\noindent Recalling that $h = h(n)$ and $k = k(n)$, we observe that 
\[
\frac{(2 + 1 / \alpha)^{h}}{2^{h} + k / \alpha}
~ = ~
\Theta \! \left( (1 + 1 / (2 \alpha))^h \right)
~ = ~
\Theta \! \left( n^{\log (1 + 1 / (2 \alpha))} \right) \enspace .
\]
Combined with Lemmas \ref{Lemma:BalancedTreeIsComplete}, \ref{lem:expression
for eff of balanced complete tree}, and \ref{lem:eff-max tree is balanced}, we
get the following corollary.

\begin{corollary} \label{corollary:EffectiveTrees}
The effect of an $n$-leaf full binary tree is $\BigO\! \left( n^{\log (1 + 1 /
(2 \alpha))} \right)$.
\end{corollary}

\noindent Now, let us return the focus to our flip forest.
Recalling that $\sum_{\text{flip trees } T} \sum_{\ell \in \mathcal{L}(T)}
\wb(\ell) = c(M^*)$, and using Corollary~\ref{corollary:flip tree cost},
we conclude that
\[
\frac{c(M_G)}{c(M^*)}
~ = ~
\BigO \! \left( \frac{\sum_{\text{flip trees } T} \Wb(T)}{\sum_{\text{flip trees
} T} \lw(T)} \right)
~ = ~
\BigO \! \left( \max_{\text{flip trees } T} \frac{\Wb(T)}{\lw(T)} \right)
~ = ~
\BigO \! \left( \max_{\text{flip trees } T} \eff(T) \right) \enspace .
\]
The desired upper bound then follows from
Corollary~\ref{corollary:EffectiveTrees}.

\LongVersion 
\subsection{Lower Bound}
\label{section:LowerBoundPoS}
\LongVersionEnd 

\newcommand{\SectionLowerBoundPoS}{
Our goal in this section is to establish the lower bound of
Theorem~\ref{theorem:PoS}.
The graph construction that lies at the heart of this lower bound, denoted
$\grt^{k}_{\alpha}$, is a direct generalization of the Reingold-Tarjan graph
$\grt^{k}$
presented in \Sec{}~\ref{section:PoA} for arbitrary values of
$\alpha$.
Specifically, the $2$-vertex graph $\grt^{1}_{\alpha}$ is identical to
$\grt^{1}$; and
the $2^{k + 1}$-vertex graph $\grt^{k + 1}_{\alpha}$ is constructed
recursively by placing $2$ disjoint instances of $\grt^{k}_{\alpha}$, each of
diameter $D^{k}_{\alpha}$, on the real line, only that this time, the spacing
between them is set to
$S^{k + 1}_{\alpha} = (1 / \alpha - \varepsilon) D^{k}_{\alpha}$,
for some sufficiently small $\varepsilon > 0$ that will be determined later
on.
This implies that
$D^{k}_{\alpha} = (2 + 1 / \alpha - \varepsilon)^{k - 1}$
and
$S^{k + 1}_{\alpha} = (1 / \alpha - \varepsilon) (2 + 1 / \alpha -
\varepsilon)^{k - 1}$.

Now let $M$ be an $\alpha$-stable matching in $\grt^{k}_{\alpha}$.
We argue that $M$ has to contain each edge $e = (x, y)$ with $w(e) = 1 /
\alpha - \varepsilon$.
Indeed, if $e \notin M$, then $e$ is $\alpha$-unstable with respect to $M$
since
$w(e) < \alpha \cdot \min \{ w(x, x'), w(y, y')\}$ for all other vertices $x',
y'$.
Given that all vertices with distance $1 / \alpha - \varepsilon$ are therefore
already matched, we can apply the same argument for each edge connecting two
adjacent vertices with edge weight $(1 / \alpha - \varepsilon) (2 + 1 /
\alpha- \varepsilon)$ and thereby conclude that these edges have to be in $M$
as well.
By repeating this argument, we end up with the unique $\alpha$-stable matching
$M$ that has to contain the edge $(x^{k}_{1}, x^{k}_{2^{k}})$ whose weight is
$D^{k}_{\alpha}$ and and all other edges whose weight differs from $1$.
Thus,
$c(M) \geq D^{k}_{\alpha} = (2 + 1 / \alpha - \varepsilon)^{k - 1}$.

On the other hand, the cost of the minimum-cost matching $M^{*}$ is not larger
than that of the matching using all weight $1$ edges, thus we can bound the
cost of $M^{*}$ as
$c(M^{*}) \leq 2^{k - 1}$.
Together, we conclude that
\begin{align}
\PoS_{\alpha}(\grt^{k}_{\alpha})
~ \geq ~ &
\frac{c(M)}{c(M^{*})} \nonumber \\
\geq ~ &
\frac{(2 + 1 / \alpha - \varepsilon)^{k - 1}}{2^{k - 1}} \nonumber \\
= ~ &
\Omega \left( 1 + \frac{1}{2 \alpha} \right)^{k -
1} \label{equation:SmallEpsilon} \\
= ~ &
\Omega \left( n^{\log \left( 1 + \frac{1}{2 \alpha} \right)} \right) \enspace
, \label{equation:ExpressN}
\end{align}
where (\ref{equation:SmallEpsilon}) holds by taking a sufficiently small
$\varepsilon$ and (\ref{equation:ExpressN}) follows by recalling that
$\grt^{k}_{\alpha}$ has $2 n = 2^{k}$ vertices.
} 
\LongVersion 
\SectionLowerBoundPoS{}
\LongVersionEnd 

\ShortVersion 
\clearpage

\pagenumbering{roman}
\appendix

\renewcommand{\theequation}{A-\arabic{equation}}
\setcounter{equation}{0}

\begin{center}
\textbf{\Large{APPENDIX}}
\end{center}

\section{Lower Bound for $\PoS_{\alpha}(2 n)$}
\label{appendix:LowerBoundPoS}
\SectionLowerBoundPoS{}

\section{Additional Proofs from \Sec{}~\ref{section:PoS}}
\label{appendix:AdditionalProofsSectionPoS}

\begin{proof}[Proof of Lemma~\ref{lemma:UnstableEdgeCreation}]
\ProofLemmaUnstableEdgeCreation{}
\end{proof}

\begin{proof}[Proof of Lemma~\ref{lemma:OutputsValidMatching}]
\ProofLemmaOutputsValidMatching{}
\end{proof}

\begin{proof}[Proof of Lemma~\ref{lemma:EdgeUnstableOnlyOnce}]
\ProofLemmaEdgeUnstableOnlyOnce{}
\end{proof}

\begin{proof}[Proof of Lemma~\ref{Lemma:NodeWeightBound}]
	\ProofLemmaNodeWeightBound{}
\end{proof}

\ShortVersionEnd 
\clearpage
\renewcommand{\thepage}{}

\bibliographystyle{abbrv}
\bibliography{references}

\end{document}